\documentclass[conference]{IEEEtran}
\usepackage{cite}
\usepackage{graphicx}
\usepackage{amsfonts}
\usepackage{caption}
\usepackage{subcaption}
\usepackage{pifont}
\usepackage{graphicx}

\usepackage{amssymb}
\usepackage{amsmath, amsthm}
\usepackage{wrapfig}
\usepackage{epsfig}
\usepackage[usenames,dvipsnames]{xcolor}
\usepackage{mathdots}
\usepackage{tikz}
\usepackage{bbm}
\ifCLASSINFOpdf
\else
\fi

\IEEEoverridecommandlockouts
\begin{document}
\newtheorem{theorem}{Theorem}
\newtheorem{example}{Example}
\newtheorem{claim}{Claim}
\newtheorem{corollary}{Corollary}
\newtheorem{conjecture}{Conjecture}
\newtheorem{definition}{Definition}
\newtheorem{lemma}{Lemma}
\newtheorem {remark}{Remark}
\newtheorem {proposition}{Proposition}
\newcommand{\proofsketch}{\vspace*{-1ex} \noindent {\emph{ Proof Sketch:} }}

\allowdisplaybreaks 
%
\title{On the Capacity Achieving Probability Measures for Molecular Receivers}

\author{Mehrdad Tahmasbi, Faramarz Fekri\\

School of Electrical and Computer Engineering\\ Georgia Institute of Technology\\ Atlanta, GA 30332 \\
Email: \{tahmasbi, fekri\}@ece.gatech.edu
\thanks{This material is based upon work supported by the National Science
Foundation under Grant No. CNS-111094}

}


%


\maketitle

\begin{abstract}
In this paper, diffusion-based molecular communication with ligand receptor receivers is studied. Information messages are assumed to be encoded via variations of the concentration of molecules. The randomness in the ligand reception process induces uncertainty in the communication; limiting the rate of information decoding. We model the ligand receptor receiver by a set of finite-state Markov channels and study the general capacity of such a receiver. Furthermore, the i.i.d. capacity of the receiver is characterized as a lower bound for the general capacity. It is also proved  that a finite support probability measure can achieve the i.i.d. capacity of the receiver. Moreover, a bound on the number of points in the support of the probability measure is obtained.
\end{abstract}


%
\IEEEpeerreviewmaketitle

\section{Introduction}
Inspired by communication paradigms in microorganisms, recent research has focused on to designing communication systems that convey messages using molecules instead of electromagnetic waves \cite{IanFernandoCristina, TadashiAndrew, HamidrezaGholamali, chris,tadashi}. These methods can be beneficial in environments such as human body where the use of electromagnetic waves is limited. Among several ways of encoding messages into molecules, the most common and well-studied method  is encoding via concentration of molecules. As such, the transmitter alters the concentration in a shared medium in order to induce distinct actions in the receiver  corresponding to different messages at the transmitter.

In this new context, like every traditional communication systems, the notion of capacity can be defined as the fundamental upper bound on the rates for which reliable communication is possible. In the past, several authors attempted to  find the capacity of the both diffusion channel as well as ligand receptors under various assumptions. In \cite{ArashMohsen2}, authors considered a memoryless model for the ligand receptors and found that Jeffery's prior is the capacity achieving distribution. They also represented some numerical results for the case with memory. In \cite{ArashMohsen}, authors used two states Markov channel model for the diffusion channel and computed its capacity using the method Shannon suggested for the telegraph channel in \cite{Shannon}. In \cite{HamidrezaGholamali}, the diffusion channel memory is modeled by an LTI-poisson process and the capacity is studied.

However, the most relevant work to what is studied in this paper is presented in \cite{PeterAndrew}. The authors in \cite{PeterAndrew} presumed that the input is the concentration of molecules at the receiver, and the receiver (viewed as a channel) is a single ligand receptor modeled by a two-state Markov channel. Using this model, \cite{PeterAndrew} showed that for a discrete input, the capacity achieving distribution is i.i.d.. Further, the distribution takes only two values, the minimum and maximum concentrations. They also proved that feedback does not increase the capacity in this setup.

In the Markov model given for the binding channel in \cite{PeterAndrew}, the state of the channel is the output in the previous epoch. In these channels, the capacity is analysed in \cite{ChenBerger} when we have feedback. On the other hand, the capacity of the channel in which the transition of the output is independent from the transition of the state of the channel is studied in \cite{MushkinDavid} and \cite{GoldsmithVaraiya}. 

In this paper, we generalized the aforementioned model from two perspectives. First, starting from a general input (continuous, discrete, or mixed), we show that there always exists a discrete distribution which achieves the capacity. In addition, we relax the assumption of the receiver having only a single receptor, and instead we allow existing of any number of the ligand receptors. Under this assumption, the number of possible outputs and states of the channel grow exponentially in term of the number of receptors. Therefore, describing the exact capacity and the capacity-achieving distribution would be difficult. Hence, instead of finding the exact expression for these objects, we tried to found some fundamental results that help us to describe them.

The paper is organized as follow: in Section \ref{sec:notations}, our notations are introduced and the physical problem including the transmitter, the diffusion channel, the and ligand receptor receiver are described. In Section \ref{sec:model}, the stochastic model used in this work is explained precisely. The main results of the paper came in Section \ref{sec:result}.

\begin{figure} 
\begin{center}
\includegraphics[scale=0.5]{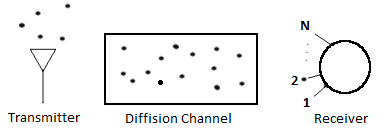}
\end{center}
\caption{\small Diffusion Channel and a Receiver with N Ligand Receptors\normalsize}
\label{fig:diffusion:channel}
\end{figure}

\section{Notations and Preliminaries} \label{sec:notations}
Assuming a sequence $(x_s, x_{s+1}, x_{s+2}, \cdots)$, we use $x_i$ to denote a single element of the sequence and $x^i$ to denote subsequence $(x_s, x_{s+1}, \cdots, x_i)$ and $x_i^j$ to show $(x_i, x_{i+1}, \cdots, x_j)$. We show the whole sequence by $x^\infty$.
Random variables are displayed by capital letters. If $X$ is a random variable $p_X(x)$ shows the probability density or mass function depending on $X$ is a continuous or discrete.
In addition, $H_2(x)$ is the binary entropy function and defined as $-x\log_2(x)-(1-x)\log_2(1-x)$. The $n^{th}$ derivative of a function $f$ is also represented by $f^{(n)}$.

There are different models for a diffusion-based communication system each of which capture some physical subtleties of the transmitter, the channel and the receiver. As depicted in Figure \ref{fig:diffusion:channel}, this paper assumes that  the transmitter is a point source emitting molecules at a controllable rate, $F(t)$. These molecules diffuse into a shared medium between the receiver and the transmitter. The diffusion phenomena can be described by Fick's Second Law:
\begin{equation}
 \frac{\partial c(x, t)}{\partial t}=D\nabla^2c(x,t)+\delta(x)F(t),
\end{equation}
where $c(x, t)$ is the concentration of molecules at position $x$ and time $t$, and $D$ is the diffusion coefficient of the medium. If the receiver is located at position $x_r$, by solving above PDE, one can find the concentration of molecules at $x_r$, $c_r(t) = c(x_r, t)$. Similar to \cite{ArashMohsen, MassimilianoIan}, in order to find $c_r(t)$, we find the impulse response (i.e., the solution when $F(t) = \delta(t))$:
\begin{equation}
h(t) = \frac{1}{4\pi Dt}\exp(-\frac{|x_r|^2}{4Dt})u(t).
\end{equation}
The solution for a general input is given as:
\begin{align} \label{eq:fick:solution}
&c_r(t) = h(t) * F(t)\nonumber\\
 &= \int^t_0 F(\tau)\frac{1}{4\pi D(t-\tau)}\exp(-\frac{|x_r|^2}{4D(t-\tau)})d\tau.
\end{align}
Additionally, the receiver is assumed as a single cell containing $N$ ligand receptors. Each receptor has two states, referred as bound ($B$) and unbound ($U$), indicating whether a ligand has bound to the receptor. To describe the stochastic behavior of a single receptor, analogous to \cite{PeterAndrew}, \emph{master equations} introduced in \cite{Desmond} can be used. If $p_ i(t)$ shows the probability that the $i^{th}$ receptor is in state $B$ at time $t$, then:
\begin{equation}\label{eq:master}
\frac{dp_i}{dt}=k_+c_r(t)(1-p_i(t))-k_-p_i(t),
\end{equation}
where $k_+$ and $k_-$ are two constants proportional to the rate of binding and unbinding reactions. Another important feature is that receptors are insensitive to the input when they are bound.

In this paper, we assume a discrete time model for the capacity study in which the input and the output of the system are sampled. More precisely, if we fix a sampling period $\Delta$, and  assume that the rate of the molecule emission by the transmitter is equal to a constant $F_n$ for time interval $n\Delta\leq t < (n+1)\Delta$ for any $n=0, 1, 2, \cdots$, the concentration at the receiver is obtained via \eqref{eq:fick:solution} as:
\begin{align*}
&c_r(m\Delta) \\
 &= \sum_{n=0}^m\int^{(n+1)\Delta}_{n\Delta} \frac{F(\tau)}{4\pi D(m\Delta-\tau)}\exp(-\frac{|x_r|^2}{4D(m\Delta-\tau)})d\tau\\
&=\sum_{n=0}^m F_nh_{m-n},
\end{align*}
where $h_n$ is defined as:
\begin{equation}
\int^{\Delta}_0 \frac{1}{4\pi D(n\Delta-\tau)}\exp(-\frac{|x_r|^2}{4D(n\Delta-\tau)})d\tau
\end{equation}

Further, we suppose that the input (over which we study capacity) is the concentration of the molecules at the receiver. The rational for our setup is that there is a one to one mapping between the rate of molecule emission at transmitter, $F^\infty$, and the sampled concentration of molecules at the receiver. It is due to the fact that $h_n$ is the integral of a positive function, and therefore it is non-zero. Hence, by induction, we can show that $F^\infty$ can be obtained uniquely from the $c_r(\Delta), c_r(2\Delta), \cdots$. Nevertheless, it should be mentioned that the constraints of the input are commonly given at the transmitter, and they should be translated to appropriate ones at receiver which may not be of the form assumed in this paper. However, this assumption helps us to find  theoretical results for the capacity of a molecular receiver.

\section{Stochastic Model of the Channel} \label{sec:model}
In this paper, we focus on the capacity  of the ligand receptor receiver. As such, to the rest of the paper, our definition of the channel only includes the receiver without diffusion channel from the transmitter to the receiver.
Since our model is discrete time, all variables of the channel at time $i\Delta$ are shown by index $i$. We call each interval of $[i\Delta, (i+1)\Delta)$ as an epoch.

\subsubsection{Input} The input of the channel denoted by $X^\infty=(X_1, X_2, \cdots)$ is the  sampled concentrations present at the receiver i.e., $X_n = c_r(n\Delta)$. These concentrations are limited to constraint $0\leq X_i \leq M$, where $M$ is a positive real number. We use $\mathcal{X}$ to denote the input alphabet, which is the continuous interval $[0, M]$.

\subsubsection{Output} The output of the channel shown by $Y^\infty=(Y_0, Y_1, Y_2, \cdots)$ is the sampled state of all $N$ receptors. In other words, if we define $Y_{i,j}$ as the state of receptor $j$ at time $i\Delta$, then $Y_i=(Y_{i,1}, \cdots, Y_{i,N})\in\{B,U\}^N$. Moreover, $Y_0$ is the initial state of the receptors and has an arbitrary distribution. We use $\mathcal{Y}$ to show the output alphabet $\{B,U\}^N$.

\subsubsection{Channel State Transition} We model the $N$ ligand receptors as $N$ Markov channels with independent state transitions.
 For the $j^{th}$ receptor at the epoch $i$, if $Y_{i,j}=B$, it would go to the unbound state at the epoch $i+1$ with probability $0<\beta<1$, where $\beta$ is independent of the input. Likewise, if $Y_{i,j}=U$, the probability that $Y_{(i+1),j}=B$ is a function of the input. This function denoted by $\alpha:\mathcal{X}\to (0,1)$  has two following properties: $\alpha$ is a strictly increasing function, and $\alpha$ is continuous.

The above two properties are sufficient to obtain all the results presented in this paper. However, it can be shown that, in our setup, we have:
\begin{equation}
\alpha(x) = \frac{k_+x}{k_-+k_+x}
\end{equation}

The following definition will help us to express our results more briefly:

\begin{definition}
If $X$ is a random variable on $\mathcal{X}$, we define the $i^{th}$ moment of $\alpha(X)$ as $m_i^X\triangleq\mathbb{E}[\alpha^i(X)]$.
For $1\leq i_1 \leq i_2$, we also define $m^X_{i_1:i_2}\triangleq(m_{i_1}^X, \cdots, m_{i_2}^X)$.
\end{definition}
\subsubsection{Capacity}
To define capacity, first we need the following definition:

\begin{definition} \textbf{Achievable Rates}\\
We say a rate $R > 0$ is achievable, if for any $\epsilon>0$, there exists an integer $n_0>0$ such that for all integers $n \geq n_0$, we can find  $K$ codewords $w_1, \cdots, w_K\in \mathcal{X}^n$ and $K$ decoding subsets $D_1, \cdots, D_K\subset\mathcal{Y}^n$ so that  $\mathbb{P}(Y^n_1\in D_i|X^n=w_i) > 1 - \epsilon$ for all $1\leq i \leq K$,.
\end{definition}
Since $|\mathcal{Y}|<\infty$, it is easy to check that the set of all achievable rates is bounded. Accordingly, we can define:
\begin{definition}\label{def:capacity} \textbf{Channel Capacity}\\
\begin{equation}
C \triangleq \sup \{R| R \text{ is an achievable rate}\}
\end{equation}
\end{definition}

\begin{definition} \label{def:iid:capacity}
\begin{equation}
C_{iid}\triangleq\lim_{n\to\infty}\max_{p_{X^n}\in \mathcal{P}^n_{iid}}\frac{1}{n}I(X^n;Y^n),
\end{equation}
where $P_{X^n}\in\mathcal{P}^n_{iid}$ if and only if $X_1, X_2, \cdots, X_n$ is an i.i.d. sequence.
\end{definition}

\begin{definition}\label{def:markov:capacity}
\begin{equation}
C_{m}\triangleq\lim_{n\to\infty}\max_{p_{X^n}\in \mathcal{P}^n_{m}}\frac{1}{n}I(X^n;Y^n),\end{equation}
where $P_{X^n} \in \mathcal{P}^n_{m}$ if and only if there exists a finite-state stationary Markov chain $S_1, S_2, \cdots, S_n$ with state space $\mathcal{S}$ and a function $\phi:\mathcal{S}\to\mathcal{X}$ such that $X_i=\phi(S_i)$ for all $1\leq i \leq n$.
\end{definition}

\section {Results} \label{sec:result}
To begin with, the next proposition characterizes some basic properties of the channel for i.i.d. input:
\begin{proposition}\label{lem:basic:entropy}
Suppose the input of the channel is an i.i.d. sequence $X^\infty$ and the corresponding output is $Y^\infty$. Then:
\begin{enumerate}
\item $Y^\infty$ is an aperiodic irreducible Markov chain.
\item There exists a unique distribution for $Y_0$ such that $Y^\infty$ forms a stationary sequence, and for any other distribution of $Y_0$, $P_{Y_n}$ tends to the stationary distribution as $n$ tends to $\infty$.
\item If $Y^\infty$ is stationary, $H(Y_i|Y_{i-1})$ depends only on $m_{1:N}^X$
\item If $Y^\infty$ is stationary, 
\begin{align}
&H(Y_i|X_i, Y_{i-1}) =\nonumber \\
& N[p_{Y(i-1),1}(B)H_2(\beta)+ p_{Y(i-1),1}(U)\mathbb{E}(H_2(\alpha(X_i)))] \label{eq:second:entropy}
\end{align}
\end{enumerate}
\end{proposition}
\begin{proof}
\begin{enumerate}
\item 
To check that $Y^\infty$ is a Markov chain, we use the fact that $X_i$ is independent of $Y_0, Y_1, \cdots, Y_{i-1}$. So:
\begin{align*}
&p_{Y_i|Y_0, \cdots, Y_{i-1}}(y_i|y_0, \cdots, y_{i-1})\\
&=\mathbb{E}_{X_i|Y_0, \cdots, Y_{i-1}}(p_{Y_i|Y_0, \cdots, Y_{i-1}, X_i}(y_i|y_0, \cdots, y_{i-1}, X_i))\\
&=\mathbb{E}_{X_i}(p_{Y_i|Y_{i-1}, X_i}(y_i|y_{i-1}, X_i))\\
&=p_{Y_i|Y_{i-1}}(y_i|y_{i-1}) .
\end{align*}
Furthermore, since $\alpha(x),\beta \in(0,1) $, the transition between every two states has positive probability. Hence, $Y^\infty$ is both irreducible and aperiodic.
\item It is a well-known fact about Markov chains.
\item 
Let $x$ be a fixed input and $y_0, y_1 \in \{U, B\}^N$ are two possible outputs. Also, define:
\begin{align*}
&N_1 = \#\{j|y_{0,j}=U, y_{1,j}=U\}\\
&N_2 = \#\{j|y_{0,j}=U, y_{1,j}=B\}\\
&N_3 = \#\{j|y_{0,j}=B, y_{1,j}=U\}\\
&N_4 = \#\{j|y_{0,j}=B, y_{1,j}=B\}.
\end{align*}
Then, $$p_{Y_i|Y_{i-1},X_i}(y_1|y_0,x)=(1-\alpha(x))^{N_1}\alpha(x)^{N_2}\beta^{N_3}(1-\beta)^{N_4}.$$
Because $X_i$ is independent of $Y_{i-1}$:
\begin{align*}
&p_{Y_i|Y_{i-1}}(y_1|y_0)= \nonumber \\
&\mathbb{E}_{X_i}((1-\alpha(X_i))^{N_1}\alpha(X_i)^{N_2}\beta^{N_3}(1-\beta)^{N_4}).
\end{align*}
Since  $N_1+N_2\leq N$, the expression inside $\mathbb{E}_{X_i}$ is a polynomial of degree at most $N$ of $\alpha(x)$. Therefore, $p_{Y_i|Y_{i-1}}(y_1|y_0)$ is a function of $m_{1:N}^X$.

Thus, if we show the stationary distribution by $\pi$, it depends only on $m_{1:N}^X$, and hence,
\begin{equation*}
H(Y_1|Y_0) = \sum_{y\in \{U, B\}|^N} \pi(y) H(Y_1|Y_0 = y).
\end{equation*}
\item 
\begin{align*}
&H(Y_i|X_i, Y_{i-1})\\
&=NH(Y_{i,1}|X_i, Y_{(i-1),1})\\
&=Np_{Y(i-1),1}(B)\mathbb{E}_{X_i}(H(Y_{i1}|X_i, Y_{(i-1),1}=B))+\\
&Np_{Y(i-1),1}(U)\mathbb{E}_{X_i}(H(Y_{i1}|X_i, Y_{(i-1),1}=U))\\
&= Np_{Y(i-1),1}(B)H_2(\beta)+\\
&Np_{Y(i-1),1}(U)\mathbb{E}_{X_i}(H_2(\alpha(X_i))).
\end{align*}
\end{enumerate}
\end{proof}

\begin{theorem} \label{thm:finite:support}
There exists a discrete distribution for the input which achieves the i.i.d. capacity of the channel, i.e., $C_{iid}$.
\end{theorem}
\begin{proof}
Suppose the i.i.d. capacity is achieved for the i.i.d. input $X^\infty$, and the corresponding output is $Y^\infty$. By Proposition \ref{lem:basic:entropy}, $Y^\infty$ is a Markov chain. We also know that given $Y_{i-1}$ and $X_i$, $Y_i$ is independent of the other components of the input and output. Hence, we have:
\begin{align*}
&I(X^n;Y^n) =H(Y^n) - H(Y^n|X^n) 
\\
&=I(Y_0;X^n) + \sum_{i=1}^nH(Y_i|Y^{i-1}) - \sum_{i=1}^nH(Y_i|Y^{i-1},X^n)\\
&=I(Y_0;X^n) + \sum_{i=1}^nH(Y_i|Y_{i-1}) - \sum_{i=1}^nH(Y_i|Y_{i-1},X_i).
\end{align*}

Initially, we suppose that $Y_0$ has the stationary distribution. Then $Y^\infty$ is a stationary process, and the above expression is equal to:
\begin{align*}
 I(Y_0;X^n) + nH(Y_1|Y_0) - nH(Y_1|Y_0,X_1)\\
\end{align*}
Thus:
\begin{align}
&\lim_{n\to \infty}\frac{1}{n}I(X^n;Y^n)\nonumber\\
&=\lim_{n\to\infty}\frac{1}{n} I(Y_0;X^n) + H(Y_1|Y_0) - H(Y_1|Y_0,X_1)\nonumber\\
&=H(Y_1|Y_0)-H(Y_1|Y_0,X_1) \label{eq:stationary:I}.
\end{align}
Note that by Proposition \ref{lem:basic:entropy}, the first term in \eqref{eq:stationary:I} is a function of $m_{1:N}^{X_1}$. We can also write the second term $H(Y_1|Y_0,X_1)$ as $\mathbb{E}_{X_1}(g(X_1))$ where $g(x) = H(Y_1|Y_0,X_1=x)$. Thus,  $\lim_{n\to \infty}\frac{1}{n}I(X^n;Y^n)$ just depends on the expected value of $N+1$ functions of $X_1$. By Lemma \ref{lem:moment:problem} in Appendix, there exists another discrete i.i.d. input $\tilde{X}^\infty$ that takes at most $N+2$ values and the expected value of these $N+1$ functions remain unchanged. Therefore, if we show the output corresponding to $\tilde{X}^\infty$ by $\tilde{Y}^\infty$, then:
\begin{align*}
\lim_{n\to \infty}\frac{1}{n}I(X^n;Y^n) = \lim_{n\to \infty}\frac{1}{n}I(\tilde{X}^n;\tilde{Y}^n)
\end{align*}
Similar to \cite{GoldsmithVaraiya} we will prove that the above limit is independent of the $Y_0$ distribution which justifies our assumption that $Y_0$ has the stationary distribution.
\end{proof}

\begin{remark}
Standard way of using the Caratheodery technique gives exponential bounds on $|\mathcal{X}|$, but from the above proof $|\mathcal{X}|$ is bounded by $N+2$ which is a significant improvement and simplifies the problem of finding the capacity.
\end{remark}

\begin{theorem} \label{thm:single:receptor}
For $N=1$, the i.i.d. capacity achieving distribution of the channel takes two values of $x_{\min}=0$ and $x_{\max}=M$.
\end{theorem}
\begin{proof}
We represent $\alpha(0)$ by $\alpha_{\min}$ and $\alpha(M)$ by $\alpha_{\max}$. Assume that $X^\infty$ is the i.i.d. sequence which achieves the i.i.d. capacity.  Define $\mu = m_1^{X_1}$. We define $\tilde{X}^\infty$ as an i.i.d. sequence where $\tilde{X}_i$ takes values $0$ and $M$ such that:
\begin{align*}
\mathbb{P}(\tilde{X}_i=0) = \frac{\alpha_{\max}-\mu}{\alpha_{\max}-\alpha_{\min}}\\
\mathbb{P}(\tilde{X}_i=M) = \frac{\mu-\alpha_{\min}}{\alpha_{\max}-\alpha_{\min}}
\end{align*}
Let $\tilde{Y}^\infty$ be the corresponding the output for input $\tilde{X}^\infty$. By Proposition \ref{lem:basic:entropy}, we know that $H(Y_1|Y_0)=H(\tilde{Y}_1|\tilde{Y}_0)$. Since $H_2(x)$ is a concave function, we have:
\begin{align}
 \frac{\alpha_{\max}-\mu}{\alpha_{\max}-\alpha_{\min}} H_2(\alpha_{\min})+ \frac{\mu-\alpha_{\min}}{\alpha_{\max}-\alpha_{\min}}H_2(\alpha_{\max})\\
=\mathbb{E}(H_2(\alpha(\tilde{X}_1)))
\end{align}
Since, $H_2$ is strictly concave function, it can be shown that
\begin{equation}
\mathbb{E}(H_2(\alpha(\tilde{X}_1)))\leq \mathbb{E}(H_2(\alpha(X_1))
\end{equation}
Therefore, by Proposition \ref{lem:basic:entropy}, $H(Y_1|X_1,Y_0)\geq H(\tilde{Y}_1|\tilde{X}_1,\tilde{Y}_0)$. Thus, we have $I(X^n;Y^n) \leq I(\tilde{X}^n; \tilde{Y}^n)$
\end{proof}

\begin{theorem} \label{thm:multi:receptor}
If $X^\infty$ is a discrete i.i.d. input that achieves the i.i.d. capacity, then $X_1$ should take at most $\frac{N+4}{2}$ values.

\end{theorem}

\begin{proof}
Similar to the proof of Theorem \ref{thm:finite:support}, we assume that $Y_0$ has the stationary distribution. Suppose that $X_1$ takes values $x_1^*<x_2^*<\cdots<x_K^*$, with probability $p_1^*, p_2^*, \cdots, p_K^*$ respectively (it can be verified that an optimal solution always exists). Additionally, let $\alpha_i^*$ be $\alpha(x_i^*)$. It is easy to verify that we can choose the distribution so that $0<p_i^*<1$ for $1\leq i \leq K$. Using \eqref{eq:second:entropy} in \eqref{eq:stationary:I},
\begin{align}
&\lim_{n\to\infty}\frac{1}{n}I(X^n;Y^n)=H(Y_1|Y_0)-H(Y_1|X_1,Y_0)=\nonumber\\
&H(Y_1|Y_0) - Np_{Y_{0,1}}(B)H_2(\beta) - Np_{Y_{0,1}}(U)\mathbb{E}[H_2(\alpha(X_1))].\label{eq:mutual}
\end{align}
Since, $H(Y_1|Y_0)$ is a function of $m^{X_1}_{1:N}$, by maintaining these momentums constant via \eqref{eq:moments}, \eqref{eq:mutual} implies that $(\alpha_1^*, \cdots, \alpha_K^*, p_1^*, \cdots, p_K^*)$ should be a solution to the following optimization problem:
\begin{align}
\text{minimize} &~ \sum_{i=1}^{K}p_iH_2(\alpha_i) \nonumber\\
\text{subject to}& ~  \sum_{i=1}^K p_i \alpha_i^j=m^{X_1}_j, for~ 1\leq j \leq N\label{eq:moments}\\
&~\sum_{i=1}^K p_i = 1\nonumber\\
&~0\leq p_i\leq1,\nonumber \\
&~ \alpha(0) \leq \alpha_i \leq \alpha(M)\nonumber.
\end{align}
The above constraint optimization problem can be written as unconstrained setup (except last two constraints),
\begin{equation}
\mathcal{L} =  \sum_{i=1}^Kp_iH_2(\alpha_i)  + \lambda_0\sum_{j=1}^Kp_j +\sum_{i=1}^N \lambda_i\sum_{j=1}^Kp_j \alpha_j^i ,
\end{equation}
by the method of Lagrange multipliers. By taking partial derivatives and evaluating at the optimal solution $(\alpha_1^*, \cdots, \alpha_K^*, p_1^*, \cdots, p_K^*)$, we get:
\begin{align}
H_2(\alpha_i^*) + \sum_{j=0}^N \lambda_j\alpha^{*j}_i = 0, ~ 1\leq i \leq K\\
p_i^*(H_2'(\alpha_i^*) + \sum_{j=0}^N j\lambda_j\alpha^{*j-1}_i) = 0, ~ 2\leq i \leq K-1.
\end{align}
Notice that only $\alpha_1^*$ and $\alpha_K^*$ can be on boundaries. Moreover, as $p_i^*\neq 0$, we have:
\begin{align}
H_2'(\alpha_i^*) + \sum_{j=0}^N j\lambda_j\alpha^{*j-1}_i = 0, ~ 2\leq i \leq K-1.
\end{align}
If we define $f:[0, 1] \to \mathbb{R}$ so that $f(x) = H_2(x) +  \sum_{j=0}^N \lambda_j\alpha^{*j}_i$, then it is a smooth function and we have:
\begin{align}
f(\alpha_i^*)=0  , ~ 1\leq i \leq K\\
f'(\alpha^*_i)=0, ~ 2\leq i \leq K-1.
\end{align}
Since for each $i$ between 1 and $K-1$, $f(\alpha_i^*)=f(\alpha_{i+1}^*)=0$, according to the mean value theorem, there should be $\alpha^{**}_i \in (\alpha_{i}^*, \alpha_{i-1}^*)$ such that $f'(\alpha_i^{**}) = 0$. Therefore, $f'$ has at least $2K-3$ roots. By Lemma \ref{lem:root2} in Appendix, the number of the roots of $f'$ is at most $N+1$. Hence, $2K-3 \leq N+1$ which would result that $K\leq \frac{N+4}{2}$
\end{proof}

\begin{theorem} \label{thm:iid:stationary}
$C_{m}=C$.
\end{theorem}
\proofsketch
Assuming F as a finite subset of $\mathcal{X}$,  we can consider the channel similar to the one introduced in this paper, except that its input alphabet is limited to $F$. Analogues to Definitions \ref{def:capacity} and \ref{def:markov:capacity}, one can define $C^F$ and $C_{m}^F$ for this channel. In \cite[Theorem 3, 4]{Blackwell}, it is proved that if the channel is irreducible, then $C^F=C_{m}^F$ (the main assumption is that the input alphabet is finite). It is easy to verify that the sufficient condition stated in \cite[Theorem 1]{Blackwell} holds for the channel. Hence, we have:
\begin{align*}
C_m = \sup_{F\subset\mathcal{X},~ |F|<\infty}C_{m}^F=\sup_{F\subset\mathcal{X},~ |F|<\infty}C^F\leq C
\end{align*}
For the converse part, we need to show that for any achievable rate $R$, we can find sequence $P_{X^n} \in \mathcal{P}_m^n$ such that $\lim_{n\to\infty}I(X^n;Y^n)$ is arbitrary close to $R$. If we construct the sequence similar to the proof of \cite[Theorem 4]{Blackwell}, it is obvious that the input has finite states and therefore that construction works here as well.
\begin{remark} \textbf{Channel with Feedback}

In \cite{PeterAndrew}, authors proved that for a single-receptor case, feedback does not change the capacity. Intuitively, when we have a coding scheme which uses feedback, the encoding function depends on the output of the channel in the previous epochs. Since the channel has Markov structure, if we go back more that one epoch, we do not get useful information. Hence, one can modify the encoding function so that it would always assume  that the previous output was $U$ (i.e., the receptor was at the unbound state). If the assumption was correct, it is similar to the feedback strategy. Otherwise, the state of the channel is $B$, and the next output is independent of the input. Thus, in both cases, the feedback strategy and the modified strategy have the same result. Therefore, every rate which is achievable via feedback can be achieved without feedback.

However, for $N>1$, the above argument fails, because we have $2^N$ distinct outputs and except for the state where all receptors are bound, the output of the channel would depend on the input.
\end{remark}

\section{Conclusion}
In this paper, the capacity achieving probability measure of a molecular receptors in concentration based communication is characterized. We tried to find some properties of the distribution that achieves the capacity. Specially, we proved that discrete distributions can achieve both general capacity and i.i.d. capacity. As such developing and efficient algorithm to find the capacity reduces to a discrete space. Moreover, we proved that the general capacity is also achieved for discrete distributions.

\appendix
\begin{lemma} \label{lem:moment:problem}
Suppose $(\Omega, \mathcal{F})$ is a measurable space and $f_1, \cdots, f_K:\Omega\to\mathbb{R}$ are $K$ measurable functions. Then for any probability measure $P$, there exists another probability measure, $P'$, with support of size at most $K + 1$ such that for any $1\leq i \leq K$:
\begin{equation}
\mathbb{E}_P[f_i] = \mathbb{E}_{P'}[f_i].
\end{equation}

\end{lemma}
\begin{lemma}\label{lem:root2}
For every $a_0, a_1, \cdots, a_n\in \mathbb{R}$, let $g(x) = H_2(x)+\sum_{i=0}^na_i  x ^i$. Then $g'(x)$ has at most $n+1$ roots in $[0, 1]$.
\end{lemma}
\proofsketch
By direct calculation, one can show that for even $k$, $H_2^{(k)}(x)$ is a strictly concave function over $[0, 1]$. Hence, either $H^{(n-1)}(x)$ or $H^{(n)}(x)$ are strictly concave. Since, every line crosses the graph of a concave function in at most 2 points, either $g^{(n)}(x)$ or $g^{(n-1)}$ has at most 2 roots. Thus, by mean value theorem, $g'(x)$ has at most $n+1$ roots.




\begin{thebibliography}{1}
\bibitem{IanFernandoCristina}
Akyildiz, I. F., Fekri, F., Sivakumar, R., Forest, C. R., \& Hammer, B. K. (2012). Monaco: fundamentals of molecular nano-communication networks. Wireless Communications, IEEE, 19(5), 12-18.
\bibitem{TadashiAndrew}
Nakano, T., Eckford, A. W., \& Haraguchi, T. (2013). Molecular communication. Cambridge University Press.

\bibitem{HamidrezaGholamali}
Arjmandi, H., Aminian, G., Gohari, A., Kenari, M. N., \& Mitra, U. (2014). Capacity of Diffusion based Molecular Communication Networks in the LTI-Poisson Model. arXiv preprint arXiv:1410.3988.
\bibitem{chris}
Song, R., et al. "Wireless signaling with identical quanta." Wireless Communications and Networking Conference (WCNC), 2012 IEEE. IEEE, 2012.
\bibitem{tadashi}
Nakano, T., Moore, M. J., Wei, F., Vasilakos, A. V., \& Shuai, J. (2012). Molecular communication and networking: Opportunities and challenges. NanoBioscience, IEEE Transactions on, 11(2), 135-148.
\bibitem{ArashMohsen2}
Einolghozati, A., Sardari, M., \& Fekri, F. (2011, October). Capacity of diffusion-based molecular communication with ligand receptors. In Information Theory Workshop (ITW), 2011 IEEE (pp. 85-89). IEEE.

\bibitem{ArashMohsen}
Einolghozati, A., Sardari, M., Beirami, A., \& Fekri, F. (2011, July). Capacity of discrete molecular diffusion channels. In Information Theory Proceedings (ISIT), 2011 IEEE International Symposium on (pp. 723-727). IEEE.
\bibitem{Shannon}
Shannon, C. E. (2001). A mathematical theory of communication. ACM SIGMOBILE Mobile Computing and Communications Review, 5(1), 3-55.
\bibitem{PeterAndrew}
Eckford, A. W., \& Thomas, P. J. (2013). Capacity of a simple intercellular signal transduction channel. arXiv preprint arXiv:1305.2245.
\bibitem {MushkinDavid}M. Mushkin and I. Bar-David, “Capacity and coding for Gilbert-Elliot channels,” IEEE Trans. Inf. Theory, vol. 35, no. 6, pp. 1277–1290, Nov. 1989
\bibitem {GoldsmithVaraiya}AJ. Goldsmith and P. Varaiya, “Capacity, mutual information, and coding for finite state Markov channels,” IEEE Trans. Inf. Theory, vol. 42, no. 3, pp. 868–886, May 1996. 
\bibitem{ChenBerger}J. Chen and T. Berger, “The capacity of finite-state markov channels  with feedback,” IEEE Trans. Info. Theory, vol. 51, pp. 780–798, Mar.  2005. 
\bibitem{MassimilianoIan}Pierobon, M., \& Akyildiz, I. F. (2012, June). Intersymbol and co-channel interference in diffusion-based molecular communication. In Communications (ICC), 2012 IEEE International Conference on (pp. 6126-6131). IEEE.
\bibitem{Desmond}Higham, D. J. (2008). Modeling and simulating chemical reactions. SIAM review, 50(2), 347-368.
\bibitem{Mulholland}
Mulholland, H. P., \& Rogers, C. A. (1958). Representation theorems for distribution functions. Proceedings of the London Mathematical Society, 3(2), 177-223.


\bibitem{Blackwell}
Blackwell, D., Breiman, L., \& Thomasian, A. J. (1958). Proof of Shannon's transmission theorem for finite-state indecomposable channels. The Annals of Mathematical Statistics, 1209-1220.

\end{thebibliography}
%

\end{document}